\documentclass[11pt,letterpaper]{article}

\usepackage{amsmath,amssymb,amsthm,amsfonts}
\usepackage{mathtools}
\usepackage{graphicx}
\usepackage{float}
\usepackage[numbers,sort&compress]{natbib}
\usepackage{setspace}
\usepackage{array}
\usepackage{booktabs}
\usepackage{longtable}
\usepackage[dvipsnames]{xcolor}
\usepackage{hyperref}
\usepackage{microtype}
\usepackage[left=1.25in,right=1.25in,top=0.90in,bottom=0.90in]{geometry}
\usepackage{titlesec}

\definecolor{citeblue}{RGB}{0,71,171}
\definecolor{linknavy}{RGB}{0,51,128}
\hypersetup{
    colorlinks=true,
    citecolor=citeblue,
    linkcolor=linknavy,
    urlcolor=linknavy,
    breaklinks=true
}

\allowdisplaybreaks[2]

\titlespacing*{\section}{0pt}{1.7em plus .3em minus .2em}{0.65em}
\titlespacing*{\subsection}{0pt}{1.3em plus .2em minus .2em}{0.45em}

\newtheoremstyle{spaciousplain}
  {0.8em}{0.7em}{\itshape}{}{\bfseries}{.}{0.55em}{}
\newtheoremstyle{spaciousdefinition}
  {0.8em}{0.7em}{\normalfont}{}{\bfseries}{.}{0.55em}{}
\newtheoremstyle{spaciousremark}
  {0.8em}{0.7em}{\normalfont}{}{\itshape}{.}{0.55em}{}

\theoremstyle{spaciousplain}
\newtheorem{theorem}{Theorem}[section]
\newtheorem{proposition}[theorem]{Proposition}
\newtheorem{lemma}[theorem]{Lemma}
\newtheorem{corollary}[theorem]{Corollary}

\theoremstyle{spaciousdefinition}
\newtheorem{definition}[theorem]{Definition}
\newtheorem{assumption}[theorem]{Assumption}
\theoremstyle{spaciousremark}

\DeclareMathOperator{\Var}{Var}
\DeclareMathOperator{\Cov}{Cov}
\DeclareMathOperator{\KL}{KL}
\DeclareMathOperator{\sgn}{sgn}
\newcommand{\E}{\mathbb{E}}
\newcommand{\R}{\mathbb{R}}
\newcommand{\Prob}{\mathbb{P}}
\newcommand{\Qprob}{\mathbb{Q}}
\newcommand{\Mprob}{\mathbb{M}}
\newcommand{\N}{\mathcal{N}}
\newcommand{\F}{\mathcal{F}}
\newcommand{\A}{\mathcal{A}}
\newcommand{\Uad}{\mathcal{U}}
\newcommand{\Cad}{\mathcal{C}}
\newcommand{\Wh}{\widehat{W}}
\newcommand{\eps}{\varepsilon}
\newcommand{\half}{\tfrac{1}{2}}

\title{Path-Space Robust Bayesian Portfolio Selection}
\author{Andy Au\footnote{Department of Industrial Engineering and Operations Research, Columbia University, New York, NY 10027, USA. Email: ka3173@columbia.edu.}}
\date{\today}

\usepackage{fancyhdr}
\numberwithin{equation}{section}
\hypersetup{pdftitle={Path-Space Robust Bayesian Portfolio Selection},pdfauthor={Andy Au}}

\begin{document}
\maketitle
\begin{abstract}
A Bayesian investor learns an unknown asset drift by Kalman--Bucy filtering and trades using the resulting estimate, but the observation model used for learning may itself be misspecified. We make the portfolio robust to an adversary who distorts the law of observed prices and pays a relative-entropy penalty. A single distortion affects both trading returns and learning, since wealth and the drift estimate respond to the same observed prices. The observation market is complete, so the fixed-target problem becomes an optimization over terminal tracking errors. Entropy duality removes the adversary's inner optimization, and the pricing budget reduces the investor's problem to one scalar equation for the budget multiplier. We show that the resulting payoff is replicable, and together with its worst-case measure forms an attained saddle. The robust position has the same sign as the non-robust one but a smaller size, and it tends to zero far from the target. A robust mean--variance investor also aims lower and starts with a smaller position. The non-robust strategy can fail badly --- with a known nonzero drift, its worst-case cost is infinite at every level of ambiguity. Learning can make this cost finite, but whether it does depends on the shape of the prior, not only on its mean and variance. The robust payoff remains worst-case optimal for any prior: it limits the error on price paths that a driftless model explains better than the reference model, which is where the non-robust error is largest.
\end{abstract}

% =============================
% Section 1: Introduction
% =============================
\newpage
\section{Introduction}
\label{sec:introduction}

A Bayesian investor estimates expected returns from prices and adjusts the portfolio as new observations arrive. The Gaussian Markowitz model of De~Franco, Nicolle, and Pham~\cite{defranco2018} turns this learning into an explicit feedback policy. But with a diffuse prior, the posterior standard deviation of the Sharpe ratio is still roughly $1/\sqrt{T}$ after $T$ years, and observing more often does not help, because the precision of a drift estimate depends on the length of the sample rather than on the number of observations (see Merton~\cite{merton1980}).

Bayesian learning handles drift uncertainty inside a specified model, but the observation model itself may be wrong. We follow Hansen and Sargent~\cite{hansensargent2001} and let an adversary change the law of observed price paths at a relative-entropy cost. The investor minimizes the resulting worst-case cost while still learning with the reference Kalman--Bucy filter. Wealth and the drift estimate are both computed from the same observed price path, so a single distortion shifts both the trading returns and the estimate.

The observation market is complete, so every terminal tracking error that meets a single pricing budget can be attained by trading, and we can optimize over terminal errors directly. Entropy duality then solves the adversary's problem in closed form for a fixed wealth target, and the investor's problem reduces to minimizing an exponential moment under the budget. We solve its first-order condition explicitly: the optimal error is a Lambert-W function of the inverse observation likelihood times a single multiplier that the budget pins down.

The budget prices errors under the pricing measure, while the loss weighs them under the reference measure, so error is cheapest on paths that the pricing measure weights more heavily than the reference model does. The robust investor limits how much error it places there. Tilting the reference measure toward large errors gives the worst-case measure, and together with the optimal error it forms an attained saddle. Replicating the optimal error gives the trading strategy, which only depends on current wealth and the filter state and stays optimal over time.

\subsection*{Main results}

Robustness makes the investor trade less. The robust position has the same sign as the non-robust one but a smaller size. It stays bounded as the gap between wealth and target grows, and it shrinks to zero far from the target, where the non-robust position keeps growing linearly. The reason is a damping effect: the robust error reacts less to new prices the larger it already is. As the gap grows, the value approaches the cost of not trading minus the entropy penalty of the pricing measure --- exactly what the adversary gets by simply removing the drift.

The exact solution also shows when the non-robust strategy breaks down. Its terminal error is proportional to the inverse observation likelihood, which is lognormal when the drift is known. Its worst-case cost is then infinite at every level of ambiguity, as long as the drift and the gap are nonzero. Even an investor who never trades does better --- the terminal error then stays at the current gap $g$ and costs only $g^2$. A Gaussian prior with positive variance keeps the cost finite, but in general whether learning helps depends on the shape of the prior, not only on its mean and variance (Section~\ref{sec:fragility}).

A mean--variance investor also chooses the target, which we find by a scalar optimization. The robust investor aims lower than a non-robust investor with the same weight on expected return, and starts with a smaller position. The investor precommits: the target is chosen once and held fixed while trading.

\subsection*{Related work}

The terminal-payoff approach follows Karatzas, Lehoczky, and Shreve~\cite{karatzaslehoczkyshreve1987} and Cox and Huang~\cite{coxhuang1989}. Lakner~\cite{lakner1995,lakner1998} and Karatzas and Zhao~\cite{karatzaszhao2001} develop it for portfolio choice with unknown drift; Xiong, Xu, and Zheng~\cite{xiongxuzheng2019} discuss completeness under partial information. Zhou and Li~\cite{zhou2000} provide the mean--variance embedding, Xiong and Zhou~\cite{xiongzhou2007} establish separation under partial information, and De~Franco, Nicolle, and Pham~\cite{defranco2018} supply our Gaussian benchmark.

The entropy penalty follows Hansen and Sargent~\cite{hansensargent2001}; see~\cite{ahs2003,hstw2006} for continuous-time formulations and~\cite{hansensargent2005} for hidden states and commitment. Bensoussan and van~Schuppen~\cite{bensoussanvanschuppen1985} show that risk-sensitive partial observation can require a cost-dependent filter. Here the loss only depends on observed wealth, so current wealth and the reference filter are enough to carry the optimal strategy forward in time. Nagai and Peng~\cite{nagaipeng2002} study risk-sensitive portfolio choice under partial information, and Schied~\cite{schied2005} solves robust portfolio problems in complete markets through a least favorable measure, the counterpart of our worst-case measure.

Maenhout~\cite{maenhout2004} scales the entropy penalty with the value function to keep portfolio rules homothetic, and Maenhout, Xing, and Balter~\cite{mxb2026} study that scaling. We keep the multiplier fixed instead, so robustness bites harder when the investor is far from the target. Pham, Wei, and Zhou~\cite{phamwei2022} study robust mean--variance selection with time-varying parameter sets. B\"auerle and Mahayni~\cite{bauerlemahayni2024} and B\"auerle and MacKay~\cite{bauerlemackay2026} combine learning with smooth ambiguity about the drift, the latter for a mean--variance criterion. We instead penalize alternative laws of the observed prices, which move trading returns and learning together.

\paragraph{Organization.}
Sections~\ref{sec:setup} and~\ref{sec:exact} develop the model and the exact solution, Section~\ref{sec:policy} compares the robust and non-robust portfolios, and Section~\ref{sec:fragility} studies the role of the prior. Section~\ref{sec:embedding} selects the mean--variance target and Section~\ref{sec:numerics} gives numerical illustrations. Proofs that are not in the main text are in the appendices, in the same order as the results.

\section{Learning, Trading, and the Budget}
\label{sec:setup}

\subsection{Market, prior, and filtering}

The risk-free rate is a constant $r_f$. The risky asset has drift $\mu$, volatility $\sigma$, and Sharpe ratio $\rho$, and its price is driven by a Brownian motion $W$:
\begin{equation}
\label{eq:price}
    \frac{dS_s}{S_s}=\mu\,ds+\sigma\,dW_s,
    \qquad \rho=\frac{\mu-r_f}{\sigma},\qquad \sigma>0.
\end{equation}
The investor observes prices but not $\rho$, and works with the whitened excess return
\begin{equation}
\label{eq:Y}
    Y_s=\frac{1}{\sigma}\left(\log\frac{S_s}{S_0}
              +\left(\frac{\sigma^2}{2}-r_f\right)s\right),
    \qquad dY_s=\rho\,ds+dW_s.
\end{equation}
The observation filtration is the usual augmentation of $\F_s=\sigma(Y_q:0\le q\le s)$, and $\Prob$ denotes the reference measure.

\begin{assumption}
\label{ass:model}
The horizon $T$ is finite, volatility is constant and strictly positive, and $\rho$ is independent of $W$ with Gaussian prior $\N(m_0,P_0)$ of mean $m_0$ and variance $P_0\ge0$. Trading is unconstrained; there are no short-sale, leverage, or bankruptcy restrictions. All trading strategies are observation-adapted.
\end{assumption}

The boundary case $P_0=0$ is the known-drift model; formulas that divide by a posterior variance are used only when that variance is positive.

\begin{proposition}[Reference filter]
\label{prop:filter}
Under Assumption~\ref{ass:model}, the conditional distribution of $\rho$ given $\F_s$ is Gaussian, with posterior variance $P_s$ and mean $m_s$ given by
\begin{equation}
\label{eq:filter}
    P_s=\frac{P_0}{1+P_0s},\qquad
    m_s=\frac{m_0+P_0Y_s}{1+P_0s},\qquad
    dm_s=P_s\,d\Wh_s,\qquad dP_s=-P_s^2\,ds,
\end{equation}
where the innovation
\begin{equation}
    \Wh_s=Y_s-\int_0^s m_q\,dq
\end{equation}
is a Brownian motion w.r.t. the filtration $(\F_s)_{s\in[0,T]}$. Also, the innovation and the observations generate the same filtration, $\F^{\Wh}=\F^Y$.
\end{proposition}

These are the Kalman--Bucy equations for a static hidden state~\cite{liptsershiryaev2001}, and the filter does not depend on the policy, as the investor is a price-taker. Since $P$ is deterministic, $m_s=m_0+\int_0^sP_q\,d\Wh_q$ is adapted to the innovation filtration, and so is $Y_s=\Wh_s+\int_0^sm_q\,dq$, which gives $\F^{\Wh}=\F^Y$.

We formulate the problem from an arbitrary deterministic time $t<T$, with current wealth $x$ and beliefs $(m,P)=(m_t,P_t)$, and write
\begin{equation}
\label{eq:shorthand}
    \tau=T-t,\qquad z=P\tau,\qquad
    P_T=\frac{P}{1+z},\qquad g=x-w,
\end{equation}
for the remaining horizon, a dimensionless measure of the learning still to come, the terminal posterior variance, and the current gap between wealth and a fixed discounted target $w$ for terminal wealth. At each starting time $t$, we restart the model from the given state and use the filtration generated by the later observation increments $Y_s-Y_t$. Expectations below are conditional on the current state, which we drop from the notation. Write $y=Y_T-Y_t$ for the observed return over the remaining horizon and
\begin{equation}
\label{eq:terminalM}
    M=m_T=\frac{m+Py}{1+z}
\end{equation}
for the terminal posterior mean. Note that $P$ is a variance and $\Prob$ a measure.

\subsection{The pricing measure and its density}

Let $\Mprob$ be the pricing measure, the observation martingale measure under which the excess return $Y_s-Y_t$ is a driftless Brownian motion; this change of measure is separate from discounting. Everything below runs through its density w.r.t. the reference measure,
\[
    Z=\frac{d\Mprob}{d\Prob}.
\]
It converts prices into reference expectations: a terminal error $G$ costs $\E_{\Mprob}[G]=\E_{\Prob}[ZG]$. Paths with large $Z$ are those the pricing measure weights more heavily than the reference model does, and the size of the optimal error will turn out to increase with $Z$.

\begin{lemma}[Pricing density]
\label{lem:density}
For $\tau>0$ and $P\ge0$, the reference and pricing measures are equivalent. If $P>0$, then
\begin{equation}
\label{eq:Zposterior}
    Z=\sqrt{1+z}\,
      \exp\!\left(\frac{m^2}{2P}-\frac{M^2}{2P_T}\right),
    \qquad 0<Z\le\sqrt{1+z}\,e^{m^2/(2P)}.
\end{equation}
If $P=0$, then $Z=\exp(-m(\Wh_T-\Wh_t)-m^2\tau/2)$ under $\Prob$.
\end{lemma}

With learning, $Z$ is bounded and depends on the path only through the terminal estimate $M$ --- it is largest when $M=0$, on paths whose data look driftless. With a known nonzero drift, $Z$ is lognormal and unbounded, and Section~\ref{sec:fragility} shows that this difference decides whether the non-robust strategy survives misspecification. The lemma follows by averaging the fixed-drift likelihood over the prior. The moments of $Z$ and the entropy $\kappa=\KL(\Mprob\Vert\Prob)$ of the pricing measure have closed forms (Lemma~\ref{lem:moments}): all positive moments are finite, $\E_{\Prob}[Z]=1$, and $\kappa<\infty$. In particular, the adversary can switch to the driftless pricing measure at the finite cost $\kappa/\eps$. We write $Z=Z_{t,T}$ when the initial time matters, and $Z_{t,s}=\E_{\Prob}[Z_{t,T}\mid\F_s]$ for the density process.

\subsection{Trading admissibility and the terminal budget}

Discounted wealth $X$ and the discounted dollar position $u$ satisfy
\begin{equation}
\label{eq:wealth}
    X_s=x+\int_t^s\sigma u_q\,dY_q,
    \qquad dX_s=\sigma u_s(m_s\,ds+d\Wh_s).
\end{equation}
For fixed $w$, the terminal error is $G=X_T-w$, the final value of the gap between wealth and target.\footnote{We use the same payoff functional under every measure: the stochastic integral is unchanged by equivalent changes of measure, and under a measure that is only a.c., we use the version on that measure's support.}

\begin{definition}[Trading and terminal classes]
\label{def:admissible}
Fix the ambiguity level $\eps=1/\theta>0$, the reciprocal of the entropy multiplier $\theta$. A strategy belongs to $\Uad_\eps(t,x,m,P)$ when it is observation-predictable,
\begin{equation}
\label{eq:admissibility}
    \E_{\Mprob}\Bigl[\int_t^T\sigma^2u_s^2\,ds\Bigr]<\infty,
    \qquad
    \E_{\Prob}[e^{\delta X_T^2}]<\infty
    \quad\text{for some }\delta>\eps.
\end{equation}
The associated terminal-error class is
\begin{equation}
\label{eq:Cclass}
    \Cad_\eps(g)=\left\{G:\ \E_{\Prob}[ZG]=g,
        \quad \E_{\Prob}[e^{\delta G^2}]<\infty
          \text{ for some }\delta>\eps\right\}.
\end{equation}
Here $G$ is a terminal random variable measurable in the future observation filtration. The next proposition identifies these budget-feasible claims with the attainable errors $X_T-w$.
\end{definition}

The trading class does not depend on the target $w$: shifting terminal wealth by a constant keeps the exponential condition, with a slightly smaller $\delta$ that is still above $\eps$.\footnote{Choose $\eps<\delta_1<\delta$ and $\eta>0$ with $\delta_1(1+\eta)<\delta$, use $(X_T-w)^2\le(1+\eta)X_T^2+(1+\eta^{-1})w^2$, and argue the same way in the other direction.}

The expected-energy condition makes $X$ a square-integrable $\Mprob$-martingale, which is what the budget needs.\footnote{With only a.s.\ finiteness of the integrated position, terminal representations could violate the initial budget (see Dudley~\cite{dudley1977}), so a terminal moment condition alone is not enough.}

\begin{proposition}[Terminal-budget equivalence]
\label{prop:budget}
Under Assumption~\ref{ass:model}, the terminal errors attained by $\Uad_\eps$ are exactly $\Cad_\eps(g)$. For each $G\in\Cad_\eps(g)$, the replicating wealth is
\begin{equation}
\label{eq:replication}
    X_s=w+\E_{\Mprob}[G\mid\F_s].
\end{equation}
Its integrand $u$ is unique up to $ds\otimes d\Mprob$-null sets.
\end{proposition}

One direction is the martingale property. For the other, the moment conditions put $G$ in $L^2(\Mprob)$, where the Brownian representation supplies the hedge. This is where completeness enters: with constraints or unspanned risk, a single budget would generally not be enough.

% =============================
% Section 3: Solving the Fixed-Target Game
% =============================

\section{Solving the Fixed-Target Game}
\label{sec:exact}

The adversary changes the law of the observed prices, but the investor keeps learning with the reference filter~\eqref{eq:filter}. Under a distorted measure, $m_s$ is still computed by that filter, but it does not have to be a posterior mean. We solve the game in two steps: entropy duality removes the adversary, and the budget of the previous section turns the investor's problem into a choice of terminal error.

A star marks an optimizer, and we do not write out its dependence on $\eps$, the state, or the target.

\subsection{The adversary and the Gibbs step}

The adversary chooses from the set of finite-entropy measures
\begin{equation}
\label{eq:adversary}
    \A=\{\Qprob\ll\Prob\text{ on }\F_T:\KL(\Qprob\Vert\Prob)<\infty\}.
\end{equation}
For $G\in\Cad_\eps(g)$, the penalized objective $J_\eps$ and the value $V^\eps$ are
\begin{equation}
\label{eq:game}
    J_\eps(G,\Qprob)=\E_{\Qprob}[G^2]-\frac{1}{\eps}\KL(\Qprob\Vert\Prob),
    \qquad
    V^\eps(t,g,m,P)=\inf_{u\in\Uad_\eps}\sup_{\Qprob\in\A}
                        J_\eps(G,\Qprob).
\end{equation}
The value is measured in units of squared discounted wealth, and increasing $\eps$ lowers the penalty attached to a given distorted measure.

\begin{lemma}[Gibbs variational identity]
\label{lem:gibbs}
If $\E_{\Prob}[e^{\delta G^2}]<\infty$ for some $\delta>\eps$, then
\begin{equation}
\label{eq:gibbs}
    \sup_{\Qprob\in\A}J_\eps(G,\Qprob)
    =\frac{1}{\eps}\log\E_{\Prob}[e^{\eps G^2}].
\end{equation}
The supremum is attained uniquely by the equivalent finite-entropy measure
\begin{equation}
\label{eq:Gibbsdensity}
    \frac{d\Qprob_G}{d\Prob}
    =\frac{e^{\eps G^2}}{\E_{\Prob}[e^{\eps G^2}]}.
\end{equation}
In particular, $\E_{\Qprob}[G^2]<\infty$ for every $\Qprob\in\A$.
\end{lemma}

\begin{proof}
This is the classical variational formula for relative entropy~\cite{dupuisellis1997}. The entropy inequality $\delta\E_{\Qprob}[G^2]\le\KL(\Qprob\Vert\Prob)+\log\E_{\Prob}[e^{\delta G^2}]$ makes $\E_{\Qprob}[G^2]$ finite for every $\Qprob\in\A$. Since $\delta>\eps$, we have $G^2e^{\eps G^2}\le Ce^{\delta G^2}$ for a constant $C$, so $\E_{\Prob}[G^2e^{\eps G^2}]<\infty$ and $\Qprob_G$ has finite entropy. For $\Qprob\in\A$,
\[
    \KL(\Qprob\Vert\Qprob_G)
    =\KL(\Qprob\Vert\Prob)-\eps\E_{\Qprob}[G^2]
        +\log\E_{\Prob}[e^{\eps G^2}]\ge0,
\]
with equality exactly when $\Qprob=\Qprob_G$. Dividing by $\eps$ and rearranging gives \eqref{eq:gibbs}.
\end{proof}

Combined with Proposition~\ref{prop:budget}, Lemma~\ref{lem:gibbs} reduces the game to minimizing an exponential moment:
\begin{equation}
\label{eq:static}
    U^\eps(t,g,m,P)
    :=\inf_{G\in\Cad_\eps(g)}\E_{\Prob}[e^{\eps G^2}],
    \qquad V^\eps=\frac{1}{\eps}\log U^\eps.
\end{equation}

\subsection{The exact terminal error}

With $2\eps c$ as the multiplier of the pricing budget, the first-order condition is
\[
    G^\star e^{\eps(G^\star)^2}=cZ.
\]
We invert this relation, choose $c$ from the budget, and then use convexity to prove global optimality.

Let $I_\eps$ be the inverse of the strictly increasing odd function $h\mapsto he^{\eps h^2}$. Solving for $h$ yields the Lambert-$W$ expression
\begin{equation}
\label{eq:inverse}
    I_\eps(a)=\sgn(a)\sqrt{\frac{W_0(2\eps a^2)}{2\eps}},\qquad I_\eps(0)=0.
\end{equation}
Here $W_0$ is the real principal Lambert function, defined by $W_0(v)e^{W_0(v)}=v$ for $v\ge0$~\cite{dlmf}. We will use
\begin{equation}
\label{eq:inversebounds}
    |I_\eps(a)|\le|a|,\qquad
    I_\eps'(a)=\frac{e^{-\eps I_\eps(a)^2}}
                    {1+2\eps I_\eps(a)^2}>0,
    \qquad W_0(v)\le\log(1+v).
\end{equation}
The last inequality follows from $1+we^w\ge e^w$ for $w\ge0$.

\begin{theorem}[Exact solution and attained saddle]
\label{thm:exact}
Under Assumption~\ref{ass:model}, fix $t<T$, $g\in\R$, $m\in\R$, $P\ge0$, and $\eps>0$. There is a unique $c=c_\eps(t,g,m,P)\in\R$ solving
\begin{equation}
\label{eq:budgetroot}
    \E_{\Prob}[ZI_\eps(cZ)]=g.
\end{equation}
It has the sign of $g$, and $c=0$ exactly when $g=0$. The unique optimal terminal error, the value, and the adversary's worst-case measure are
\begin{align}
    G^\star&=I_\eps(cZ),\qquad G^\star e^{\eps(G^\star)^2}=cZ,\label{eq:Hstar}\\
    U^\eps&=\E_{\Prob}[e^{\eps(G^\star)^2}],\qquad
    V^\eps=\frac{1}{\eps}\log U^\eps,\label{eq:exactvalue}\\
    \frac{d\Qprob^\star}{d\Prob}
    &=\frac{e^{\eps(G^\star)^2}}{U^\eps}.\label{eq:Qstar}
\end{align}
The terminal error has every square-exponential moment under $\Prob$, and its replication is admissible. The pair $(G^\star,\Qprob^\star)$ is a fixed-measure saddle:
\begin{equation}
\label{eq:saddle}
    J_\eps(G^\star,\Qprob)\le J_\eps(G^\star,\Qprob^\star)
        \le J_\eps(G,\Qprob^\star)
\end{equation}
for all $G\in\Cad_\eps(g)$ and $\Qprob\in\A$. More precisely,
\begin{equation}
\label{eq:staticgap}
    J_\eps(G,\Qprob^\star)-J_\eps(G^\star,\Qprob^\star)
       =\E_{\Qprob^\star}[(G-G^\star)^2].
\end{equation}
\end{theorem}

\begin{proof}
The budget map $c\mapsto\E_{\Prob}[ZI_\eps(cZ)]$ is odd and strictly increasing, and it is continuous by dominated convergence, since $|ZI_\eps(cZ)|\le|c|Z^2$ on bounded $c$-intervals. As $c\uparrow\infty$, its integrand increases pointwise to infinity because $Z>0$, so the map is unbounded by monotone convergence. By oddness it is onto $\R$, and existence, uniqueness, and the sign statement follow.

If $P>0$, then $Z$ is bounded, and hence so is $G^\star$. In general, for every $\delta>0$, the bounds \eqref{eq:inversebounds} imply
\begin{equation}
\label{eq:expbound}
    e^{\delta(G^\star)^2}
    \le(1+2\eps c^2Z^2)^{\delta/(2\eps)}.
\end{equation}
The right side is integrable by Lemma~\ref{lem:moments}, also at $P=0$, so $G^\star\in\Cad_\eps(g)$ and its replication is admissible by Proposition~\ref{prop:budget}.

The function $h\mapsto e^{\eps h^2}$ is strictly convex, and its derivative at $G^\star$ is $2\eps G^\star e^{\eps(G^\star)^2}=2\eps cZ$. For any feasible $G$,
\[
    e^{\eps G^2}\ge e^{\eps(G^\star)^2}
                    +2\eps cZ(G-G^\star).
\]
When we take expectations, the last term vanishes because $G$ and $G^\star$ have the same budget. This proves optimality, and uniqueness follows from strict convexity.

The worst-case measure has finite entropy by \eqref{eq:expbound} and Lemma~\ref{lem:gibbs}. Also,
\begin{equation}
\label{eq:weightedFOC}
    G^\star\frac{d\Qprob^\star}{d\Prob}=\frac{c}{U^\eps}Z.
\end{equation}
Consequently,
\[
    \E_{\Qprob^\star}[G^\star(G-G^\star)]
    =\frac{c}{U^\eps}\E_{\Prob}[Z(G-G^\star)]=0,
\]
and expanding the square proves \eqref{eq:staticgap}. The other saddle inequality is the optimality of the Gibbs measure, and all cross terms are integrable by the entropy bound in Lemma~\ref{lem:gibbs}.
\end{proof}

The Lambert form explains most of what follows. The non-robust error of Section~\ref{sec:value} grows linearly in $Z$. The robust error looks the same when $|cZ|$ is small, since $G^\star\approx cZ$, but for large $|cZ|$ it grows only like the square root of a logarithm:
\[
    |G^\star|\sim\sqrt{\frac{\log|cZ|}{\eps}}\qquad\text{as }|cZ|\to\infty.
\]
The robust investor therefore limits the error where $Z$ is large, on paths that a driftless model explains better than the reference model.

By Proposition~\ref{prop:budget}, the terminal saddle yields an attained strategy--measure saddle $(u^\star,\Qprob^\star)$, where $u^\star$ replicates $G^\star$.

\subsection{From terminal error to trading}
\label{sec:dynamic}

Replication turns $G^\star$ into a predictable strategy. The next result shows that, with $w$ fixed, its continuation remains optimal from the states the strategy visits.

\begin{theorem}[Replication and continuation]
\label{thm:continuation}
Fix an initial state and let $c$ and $G^\star=I_\eps(cZ_{t,T})$ be as in Theorem~\ref{thm:exact}. Set
\begin{equation}
\label{eq:Greplication}
    X_s^\star-w=\E_{\Mprob}[G^\star\mid\F_s],\qquad t\le s\le T.
\end{equation}
For every deterministic $s<T$, the conditional problem at $(s,X_s^\star-w,m_s,P_s)$ has terminal optimizer $G^\star$ and multiplier
\begin{equation}
\label{eq:restartc}
    c_s=cZ_{t,s}=c_\eps(s,X_s^\star-w,m_s,P_s).
\end{equation}
In particular, the feedback needs no additional multiplier state. The process
\begin{equation}
\label{eq:Umartingale}
    U^\eps(s,X_s^\star-w,m_s,P_s)
       =\E_{\Prob}[e^{\eps(G^\star)^2}\mid\F_s]
\end{equation}
is a true martingale under the reference measure.
\end{theorem}

Density factorization and uniqueness of the budget root determine $c_s$ from current wealth and beliefs.

To express the adversary in feedback form, we use the dynamic description of a change of measure. Under an equivalent measure with finite entropy, wealth and the filter share a single drift distortion $\varphi$:
\begin{equation}
\label{eq:distorted}
\begin{aligned}
    dX_s&=\sigma u_s(m_s+\varphi_s)\,ds+\sigma u_s\,d\Wh_s^{\Qprob},\\
    dm_s&=P_s\varphi_s\,ds+P_s\,d\Wh_s^{\Qprob},\qquad
    \KL(\Qprob\Vert\Prob)=\half\E_{\Qprob}\Bigl[\int_t^T\varphi_s^2\,ds\Bigr].
\end{aligned}
\end{equation}
This is Girsanov's theorem under a finite entropy condition~\cite{leonard2012}. With the sign convention $+\int\varphi\,d\Wh$ in the density, $\Wh_s^{\Qprob}=\Wh_s-\Wh_t-\int_t^s\varphi_q\,dq$ is a Brownian motion under $\Qprob$.\footnote{The game also allows non-equivalent measures, but we use this representation only for the worst-case measure, which is equivalent to $\Prob$.}

\begin{theorem}[Feedback formulas]
\label{thm:feedback}
At any state with $t<T$, $P\ge0$, and $\eps>0$, let $G^\star=I_\eps(cZ)$. The optimal current position and the adversary's drift distortion are
\begin{align}
    u^\star(t,g,m,P)
      &=-\frac1\sigma\E_{\Mprob}\!\left[\frac{MG^\star}{1+2\eps(G^\star)^2}\right],\label{eq:uexact}\\
    \varphi^\star(t,g,m,P)
      &=-(1+z)\,\E_{\Qprob^\star}\!\left[M\,\frac{2\eps(G^\star)^2}{1+2\eps(G^\star)^2}\right].\label{eq:phiexact}
\end{align}
At $P=0$, these formulas use $M=m$, while $G^\star$ remains a function of the random observation increment. The density process of $\Qprob^\star$ is a true martingale and has finite expected distortion energy.
\end{theorem}

Since $dX_s=\sigma u_s\,dY_s$ under $\Mprob$, the position is the Brownian coefficient of the conditional claim price, divided by $\sigma$. With the initial multiplier held fixed,
\[
    \partial_yG^\star=-\frac{MG^\star}{1+2\eps(G^\star)^2},
\]
and averaging this derivative conditionally gives the position. The adversary's drift pushes against the terminal estimate $M$, weighted toward paths with large errors.

\begin{corollary}[Direction of the distortion]
\label{cor:distortion}
At any state with $t<T$, $P\ge0$, and $\eps>0$,
\begin{equation}
\label{eq:distortionbounds}
    m\varphi^\star\le0,\qquad |\varphi^\star|\le |m|,
\end{equation}
and both inequalities are strict when $m\ne0$ and $g\ne0$.
\end{corollary}

The adversary therefore pushes the observed drift toward zero but never past it.

The feedback $\varphi^\star$ represents $\Qprob^\star$ along the optimal wealth process. Evaluated along a competitor's wealth, the same feedback would describe a different measure, which is why the saddle comparison holds $\Qprob^\star$ fixed.

% =============================
% Section 4: How Robustness Changes the Portfolio
% =============================

\section{How Robustness Changes the Portfolio}
\label{sec:policy}

We now compare the robust investor with the non-robust one at the same time, wealth, beliefs, and target. If the investor is already at the target ($g=0$) or the current estimate is zero ($m=0$), both positions are zero, so the interesting case is $mg\ne0$. One mechanism drives the results below. A position is the average sensitivity of the terminal error to the observed return, and for the two optimal errors these sensitivities are
\[
    \partial_yG_0=-MG_0,\qquad
    \partial_yG^\star=-\frac{MG^\star}{1+2\eps(G^\star)^2}.
\]
The robust sensitivity carries a damping factor that shrinks as the error grows. This is why the robust position keeps the sign of the non-robust one but is smaller, and why it fades far from the target, where errors are large. Theorem~\ref{thm:attenuation} and Corollary~\ref{cor:boundedu} make this precise.

\subsection{The non-robust benchmark}
\label{sec:value}

Without the adversary, the investor minimizes $\E_{\Prob}[G^2]$ under the same budget, over square-integrable strategies. By Cauchy--Schwarz,
\[
    g^2=\E_{\Prob}[ZG]^2\le\E_{\Prob}[Z^2]\,\E_{\Prob}[G^2],
\]
with equality exactly when $G$ is proportional to $Z$. So the optimal error is a multiple of $Z$, and optimal trading keeps only a fraction $A$ of the no-trade cost $g^2$:
\begin{equation}
\label{eq:baseline}
    G_0=AgZ,\qquad V_0=Ag^2,\qquad
    A=\frac{1}{\E_{\Prob}[Z^2]}
     =\frac{\sqrt{1+2z}}{1+z}
       \exp\!\left(-\frac{m^2\tau}{1+2z}\right).
\end{equation}
This is the formula of Theorem~\ref{thm:exact} at $\eps=0$, where $I_\eps$ becomes the identity. Its feedback \eqref{eq:u0} coincides with the Bayesian Markowitz policy of De~Franco, Nicolle, and Pham~\cite{defranco2018}; Appendix~\ref{proof:baseline} makes the class precise and checks their admissibility condition.

\subsection{Smaller positions and retreat from the target}

\begin{theorem}[Exact sign and size comparison]
\label{thm:attenuation}
Under Assumption~\ref{ass:model}, fix $t<T$, $P\ge0$, $\eps>0$, and a target $w$. The non-robust current policy is
\begin{equation}
\label{eq:u0}
    u_0(t,g,m,P)=-\frac{mg}{\sigma(1+2z)}.
\end{equation}
If $mg\ne0$, then
\begin{equation}
\label{eq:attenuation}
    0<\frac{u^\star}{u_0}<1,
\end{equation}
that is, the robust position has the same sign as the non-robust one and a smaller size.
\end{theorem}

The proof also uses the common budget of the two errors.

Near maturity, the ratio in \eqref{eq:attenuation} becomes explicit: as $t\uparrow T$ with $(g,m,P)$ fixed,
\begin{equation}
\label{eq:terminalselectors}
    \frac{u^\star}{u_0}\longrightarrow\frac{1}{1+2\eps g^2},
\end{equation}
which is the damping factor at the current gap.

\begin{corollary}[Bounded exposure and retreat far from target]
\label{cor:boundedu}
At fixed $(t,m,P,\eps)$ with $t<T$ and $\eps>0$,
\begin{equation}
\label{eq:uniformu}
    |u^\star(t,g,m,P)|\le
       \frac{\E_{\Mprob}[|M|]}{2\sigma\sqrt{2\eps}}
    \quad\text{for every }g,\qquad
    \lim_{|g|\to\infty}u^\star(t,g,m,P)=0.
\end{equation}
\end{corollary}

\begin{proof}
Since $|h|/(1+2\eps h^2)$ has maximum $1/(2\sqrt{2\eps})$, the bound follows from \eqref{eq:uexact}. As $g\to\infty$, the budget root satisfies $c\to\infty$. Since $Z>0$, it follows that $G^\star=I_\eps(cZ)\to\infty$ pointwise and $G^\star/(1+2\eps (G^\star)^2)\to0$. Dominated convergence with the bound $|M|/(2\sqrt{2\eps})$ yields the limit, and odd symmetry handles negative $g$.
\end{proof}

The limit compares states with the same beliefs; along a single trajectory, beliefs move as well.

\subsection{Value bounds}

\begin{proposition}[Value bounds and limits]
\label{prop:bounds}
For $\eps>0$ and the states of Theorem~\ref{thm:exact},
\begin{equation}
\label{eq:valuebounds}
    \max\{Ag^2,\ g^2-\kappa/\eps\}\le V^\eps\le g^2.
\end{equation}
The value is even in $g$ and in $m$, nondecreasing in $\eps$, and zero at $g=0$. For a fixed state,
\begin{equation}
\label{eq:limitsvalue}
    \lim_{\eps\downarrow0}V^\eps=Ag^2,\qquad
    \lim_{|g|\to\infty}\bigl(g^2-V^\eps\bigr)=\frac{\kappa}{\eps}.
\end{equation}
In particular, $V^\eps/g^2\to1$ as $|g|\to\infty$. For $g\ne0$, the exact scaling is
\begin{equation}
\label{eq:scaling}
    V^\eps(t,g,m,P)=g^2V^{\eps g^2}(t,1,m,P).
\end{equation}
\end{proposition}

The upper bound is the cost of not trading. The first lower bound is the non-robust optimum. For the second, let the adversary pick the pricing measure $\Mprob$: by Jensen's inequality, every admissible strategy then pays at least $g^2-\kappa/\eps$. The limit as $|g|\to\infty$ says that this simple adversary is as good as the optimal one when the gap is large.

% =============================
% Section 5: Entropic Fragility and the Role of the Prior
% =============================

\section{Entropic Fragility and the Role of the Prior}
\label{sec:fragility}

We now evaluate the non-robust strategy under the robust criterion. By the Gibbs identity (Lemma~\ref{lem:gibbs}), the worst-case cost of a terminal error $G$ is
\[
    \frac{1}{\eps}\log\E_{\Prob}[e^{\eps G^2}],
\]
possibly infinite. The table compares a known nonzero drift, a Gaussian prior with positive variance, and the two-point prior $\half\delta_{1/2}+\half\delta_{3/2}$, in each case with $t<T$, $\eps>0$, and $g\ne0$.
\begin{center}
\begin{tabular}{@{}lccc@{}}
\toprule
Worst-case cost & Known drift & Gaussian prior & Two-point prior \\
\midrule
Non-robust benchmark & $+\infty$ & finite, not optimal & $+\infty$ \\
Robust optimizer & finite, optimal & finite, optimal & finite, optimal \\
\bottomrule
\end{tabular}
\end{center}

\begin{theorem}[Gaussian learning and known-drift fragility]
\label{thm:fragility}
Fix $t<T$ and $g\ne0$. If $P>0$, then the non-robust benchmark has finite worst-case cost at every $\eps>0$, and its worst-case measure is attained. If $P=0$ and $m\ne0$, then
\begin{equation}
\label{eq:infinitecost}
    \frac{1}{\eps}\log\E_{\Prob}[e^{\eps G_0^2}]=+\infty
    \quad\text{for every }\eps>0.
\end{equation}
In that case $G_0$ is not admissible for the positive-$\eps$ problem, although it is admissible and optimal for the non-robust quadratic problem.
\end{theorem}

The mechanism is the tail of the benchmark loss. For $g\ne0$, the loss is bounded when $P>0$, but at a known nonzero drift it is lognormal, and the exponential of a lognormal variable has infinite mean at every positive scale. The proof also checks the divergence directly in the game, with a sequence of finite-entropy measures.

\subsection{A likelihood criterion beyond Gaussian priors}

The Gaussian prior helps because it makes $Z$ bounded. To see which other priors do the same, we replace the current Gaussian posterior, in this subsection only, by any distribution $\nu$ on $\R$, independent of future observation noise. We keep $\tau>0$, the complete observation market, and the same admissibility convention, and define the mixture likelihood and pricing density by
\begin{equation}
\label{eq:generallikelihood}
    L_\tau(y)=\int_{\R}e^{\rho y-\rho^2\tau/2}\nu(d\rho),
    \qquad Z(y)=\frac1{L_\tau(y)},\qquad y\sim\N(0,\tau)\text{ under }\Mprob.
\end{equation}
By Lemma~\ref{lem:generalprior}, $Z$ still has all positive moments, so the terminal optimizer, replication, and saddle carry over. The feedback formulas of Theorem~\ref{thm:feedback}, by contrast, use the Gaussian form of the posterior and would have to be recomputed.

\begin{proposition}[Bounded inverse likelihood]
\label{prop:boundedprior}
Under the hypotheses of this subsection, $Z$ is essentially bounded if and only if either
\begin{equation}
\label{eq:supportcriterion}
    \nu(\{0\})>0,
    \qquad\text{or}\qquad
    \nu(({-\infty},0))>0\ \text{and}\ \nu((0,\infty))>0.
\end{equation}
Whenever this holds, $G_0=gZ/\E_{\Prob}[Z^2]$ is bounded and has finite worst-case cost at every $\eps>0$.
\end{proposition}

An atom at zero bounds the likelihood from below, and prior mass of both signs makes it grow in both tails; otherwise it vanishes in one tail.

\begin{theorem}[Finite-support fragility classification]
\label{thm:finiteprior}
Suppose $\nu=\sum_{i=1}^np_i\delta_{\rho_i}$, where $p_i>0$, $\sum_i p_i=1$, and $g\ne0$. Then
\begin{equation}
\label{eq:finiteclassification}
    \E_{\Prob}[e^{\eps G_0^2}]<\infty\ \text{for every }\eps>0
    \quad\Longleftrightarrow\quad
    0\in[\min_i\rho_i,\max_i\rho_i].
\end{equation}
If zero is outside that interval, the expectation is infinite for every positive $\eps$.
\end{theorem}

For a one-sided finite prior, the inverse likelihood grows exponentially in the opposite observation tail, and the exponential of its square is too large for the Gaussian tail to integrate.

The priors $\N(1,1/4)$ and $\half\delta_{1/2}+\half\delta_{3/2}$ have the same mean and variance and both allow learning --- yet only the Gaussian prior gives the benchmark a finite cost. Fragility is decided by the tails of the likelihood, which the first two moments of the prior do not control.

% =============================
% Section 6: Precommitment Mean--Variance Choice
% =============================

\section{Precommitment Mean--Variance Choice}
\label{sec:embedding}

In the mean--variance problem, the investor chooses the target $w$ as well as the strategy. Zhou and Li~\cite{zhou2000} turn a mean--variance problem into a family of tracking problems, one for each target, and we adapt their embedding to account for the adversary's response to the target. The target is chosen once and held fixed.

\subsection{The fixed-control target saddle}

For admissible terminal wealth $X$, a candidate target $v$, and a weight $\lambda$ on expected return, define the entropic tracking cost $F_X$ and the robust mean--variance criterion $C_\lambda$:
\begin{align}
    F_X(v)&=\frac{1}{\eps}\log\E_{\Prob}[e^{\eps(X-v)^2}],\label{eq:FX}\\
    C_\lambda(X)&=\sup_{\Qprob\in\A}
    \left\{\Var_{\Qprob}(X)-2\lambda\E_{\Qprob}[X]
                         -\frac{1}{\eps}\KL(\Qprob\Vert\Prob)\right\}.
                         \label{eq:Clambda}
\end{align}

\begin{proposition}[Robust target embedding]
\label{prop:embedding}
If $\E_{\Prob}[e^{\delta X^2}]<\infty$ for some $\delta>\eps$, then $F_X$ is finite and twice continuously differentiable, with
\begin{equation}
\label{eq:Fderivatives}
    F_X'(v)=2(v-\E_{\Qprob_v}[X]),\qquad
    F_X''(v)=2+4\eps\Var_{\Qprob_v}(X)\ge2,
\end{equation}
where $\Qprob_v$ is the Gibbs measure of $(X-v)^2$. For every $\lambda\in\R$,
\begin{equation}
\label{eq:embedding}
    C_\lambda(X)=\lambda^2+\min_{v\in\R}\{F_X(v)-2\lambda v\}.
\end{equation}
The minimizing target is unique and satisfies $v=\E_{\Qprob_v}[X]+\lambda$. The corresponding Gibbs measure attains the supremum in \eqref{eq:Clambda}.
\end{proposition}

The target and its Gibbs measure are mutual best responses, and the proof uses this saddle to justify the representation without a minimax theorem. The linear mean term is essential: without it, minimizing over $w$ would select $w=x$, and the investor would not trade.

\subsection{The optimal target and exposure}

To choose the target, we need the marginal cost of the gap.

\begin{proposition}[Marginal cost of the gap]
\label{prop:curvature}
At fixed $(t,m,P,\eps)$ with $t<T$ and $\eps>0$, $V^\eps$ is smooth and strictly convex in $g$, with
\begin{equation}
\label{eq:Vg}
    V_g^\eps=\frac{2c}{U^\eps},\qquad
    V_{gg}^\eps\big|_{g=0}=2A.
\end{equation}
For $g>0$,
\begin{equation}
\label{eq:slopebounds}
    2Ag\le V_g^\eps\le2g.
\end{equation}
The lower bound is strict when $A<1$, that is, unless $P=0$ and $m=0$.
\end{proposition}

The first identity says that the marginal cost of the gap is the budget multiplier, rescaled by $U^\eps$. It comes from differentiating the budget equation, and the slope bounds come from Cauchy--Schwarz together with a covariance inequality.

\begin{theorem}[Optimal target and mean--variance saddle]
\label{thm:target}
Under Assumption~\ref{ass:model}, fix $t<T$, initial wealth $x$, and $\lambda>0$. The precommitment robust mean--variance value is
\begin{equation}
\label{eq:MVvalue}
    \inf_{u\in\Uad_\eps}C_\lambda(X_T^u)
    =\lambda^2-2\lambda x+
       \min_{g\in\R}\{V^\eps(t,g,m,P)+2\lambda g\}.
\end{equation}
The unique minimizing gap $g_\lambda^\star<0$ satisfies
\begin{equation}
\label{eq:targetFOC}
    V_g^\eps(t,g_\lambda^\star,m,P)=-2\lambda.
\end{equation}
Let $w_\lambda^\star=x-g_\lambda^\star$. Then
\begin{equation}
\label{eq:targetbounds}
    \lambda\le|g_\lambda^\star|\le\frac{\lambda}{A},\qquad
    x+\lambda\le w_\lambda^\star\le x+\frac{\lambda}{A}=w_\lambda^0.
\end{equation}
When $A<1$, both upper bounds are strict. Replicating the exact terminal error at this target gives an attained fixed-measure saddle for \eqref{eq:Clambda}. Its initial position obeys
\begin{equation}
\label{eq:MVexposure}
    |u^\star(t,g_\lambda^\star,m,P)|
        \le |u_0(t,g_\lambda^0,m,P)|,
    \qquad g_\lambda^0=-\lambda/A.
\end{equation}
The inequality is strict when $m\ne0$.
\end{theorem}

The target is pinned down by strict convexity and the slope bounds. For the mean--variance saddle, the proof shows that any budget-feasible change of wealth adds its own variance under the fixed Gibbs measure.

Robustness therefore reduces exposure through two channels: the robust target calls for less benchmark exposure, and the robust policy then shrinks that exposure. Both comparisons hold $\lambda$ fixed rather than the achieved expected return.

At the selected target, terminal wealth has mean $w_\lambda^\star-\lambda$ under the worst-case measure $\Qprob^\star$. Because the worst-case measure changes with $\lambda$, the resulting mean--variance curve collects points under different measures rather than tracing a frontier under a single one.

% =============================
% Section 7: Numerical Illustrations
% =============================

\section{Numerical Illustrations}
\label{sec:numerics}

We use the illustrative values $\tau=1$, $m=0.5$, $P=1$, and $\sigma=1$, for which $A\simeq0.797$ and $\kappa\simeq0.159$. Each point of Figure~\ref{fig:exact-comparison} needs one scalar root and a few Gaussian integrals, and Figure~\ref{fig:optimal-target} adds an outer search for the target.\footnote{We use 512-node Gauss--Hermite quadrature and compare it with 256 nodes at every displayed point. The changes in all plotted quantities, and all budget and target residuals, are below $3\times10^{-12}$. The budget equation is solved in $\log|c|$, and $I_\eps$ is evaluated through the Wright omega function, $W_0(e^x)=\omega(x)$, to avoid overflow.}

Figure~\ref{fig:exact-comparison} fixes $\eps=0.15$ and varies the gap. The robust position peaks near $|g|\approx2.4$ and then retreats, while the non-robust position grows linearly, and the normalized value climbs from $A$ toward $1$.

\begin{figure}[htbp]
\centering
\includegraphics[width=\textwidth]{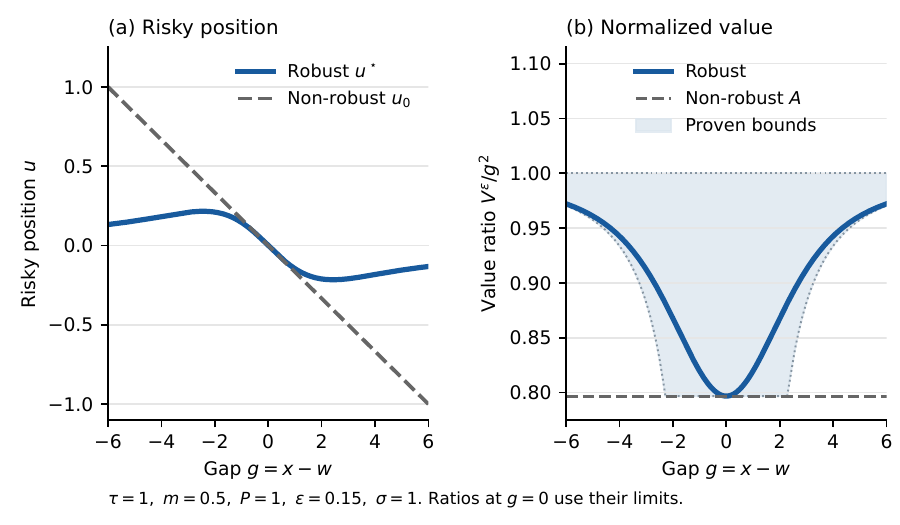}
\caption{Robust policy and normalized value at $\eps=0.15$. Dashed lines give the non-robust benchmark. The shaded band is the interval between $\max\{A,1-\kappa/(\eps g^2)\}$ and $1$ from Proposition~\ref{prop:bounds}. At $g=0$ the value ratios use their continuous extensions.}
\label{fig:exact-comparison}
\end{figure}

Figure~\ref{fig:optimal-target} varies $\eps$ in the mean--variance problem with $\lambda=1$. The target and the initial exposure both fall steadily. Theorem~\ref{thm:target} proves the comparison with $\eps=0$, and the figure shows the decline continuing between positive levels. Over the plotted range, the worst-case measure also stays closer to $\Prob$ than the driftless pricing measure does: its entropy is at most about $0.104$, against $\kappa\simeq0.159$. Beyond the plotted range it keeps rising toward $\kappa$, reaching about $0.139$ at $\eps=5$.

\begin{figure}[htbp]
\centering
\includegraphics[width=\textwidth]{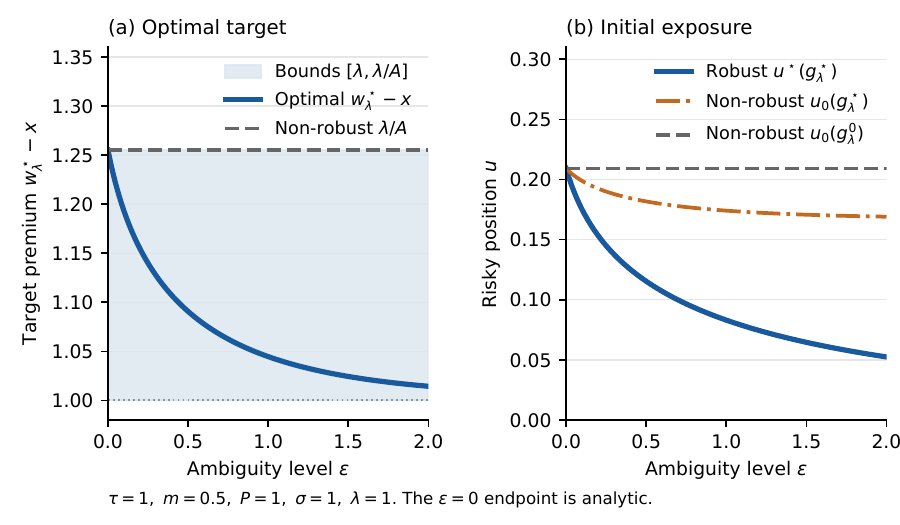}
\caption{Optimal precommitment target and initial exposure as $\eps$ varies, with $\lambda=1$. The target premium lies in $[\lambda,\lambda/A]$. The middle exposure curve evaluates the non-robust feedback at the robust target, separating target selection from the robust feedback. The comparison fixes $\lambda$, not the achieved expected return.}
\label{fig:optimal-target}
\end{figure}

% =============================
% Section 8: Conclusion
% =============================
\newpage
\section{Conclusion}
\label{sec:conclusion}

We have shown that, in a complete observation market, robust Bayesian tracking reduces to a terminal-error problem with one pricing budget. Entropy duality and a scalar multiplier determine the optimal payoff, replication turns it into a predictable strategy, and the first-order condition yields a fixed-measure saddle.

Under Gaussian learning, robustness reduces exposure at the same state and target, and the position tends to zero as the gap grows. Whether the non-robust benchmark survives misspecification depends on the tail of the pricing density --- and hence on the shape of the prior, not only on its mean and variance. In the mean--variance problem, the precommitted robust target preserves the comparison of initial exposures.

With portfolio constraints or unspanned risk, a single budget no longer describes the attainable payoffs, and the exact solution here can serve as a benchmark for those extensions.

\section*{Acknowledgments and disclosure}
The author used OpenAI Codex and Anthropic Claude for editing, proof checks, numerical code and figures, and references. An earlier version of this paper (arXiv:2606.24212v1) had no closed-form solution; the Lambert-W solution and the saddle argument were found with these tools. The author is responsible for the manuscript and for independently verifying its arguments, computations, and citations.

\clearpage
\appendix
\setstretch{1.10}
\setlength{\parskip}{0.25em plus 0.08em minus 0.04em}

\section{The pricing density and the budget}
\label{app:prelim-proofs}

\subsection{The pricing density}
\label{proof:density}

\begin{proof}[Proof of Lemma~\ref{lem:density}]
For a fixed drift $\rho$, the likelihood of the increment $y$ relative to $\Mprob$ is $\exp(\rho y-\rho^2\tau/2)$. Averaging it over the prior $\N(m,P)$ gives the mixture likelihood $L=d\Prob/d\Mprob$. It is finite and strictly positive, with unit $\Mprob$-expectation by Tonelli's theorem, so the two measures are equivalent, and completing the square gives
\begin{equation}
\label{eq:Z}
    Z=\frac1L=\sqrt{1+z}\,
       \exp\!\left(\frac{m^2\tau-2my-Py^2}{2(1+z)}\right).
\end{equation}
Substituting \eqref{eq:terminalM} gives \eqref{eq:Zposterior}. At $P=0$, \eqref{eq:Z} reads $Z=e^{-my+m^2\tau/2}$, and $y=m\tau+\Wh_T-\Wh_t$ under $\Prob$.
\end{proof}

\subsection{Gaussian moments}
\label{proof:moments}

\begin{lemma}[Gaussian moments and entropy]
\label{lem:moments}
For every real $n\ge0$, let $\mu_n=\E_{\Prob}[Z^n]$. Then
\begin{align}
    \mu_n&=\frac{(1+z)^{n/2}}{\sqrt{1+nz}}
       \exp\!\left(\frac{n(n-1)m^2\tau}{2(1+nz)}\right),\label{eq:moments}\\
    \E_{\Prob}[MZ^n]&=\frac{m\mu_n}{1+nz}.\label{eq:mixedmoments}
\end{align}
In particular all positive moments are finite, $\mu_1=1$, and
\begin{equation}
\label{eq:kappa}
    \kappa:=\KL(\Mprob\Vert\Prob)
    =\frac12\left[\log(1+z)-\frac{z}{1+z}
                        +\frac{m^2\tau}{1+z}\right]<\infty.
\end{equation}
The terminal posterior distributions are
\begin{equation}
\label{eq:Mlaws}
    M\sim
    \begin{cases}
     \N\!\left(m,\dfrac{P^2\tau}{1+z}\right),&\Prob,\\[3pt]
     \N\!\left(\dfrac{m}{1+z},\dfrac{P^2\tau}{(1+z)^2}\right),&\Mprob.
    \end{cases}
\end{equation}
The normal distributions are degenerate when $P=0$.
\end{lemma}

\begin{proof}[Proof of Lemma~\ref{lem:moments}]
Under $\Prob$, $y\sim\N(m\tau,\tau(1+z))$; under $\Mprob$, $y\sim\N(0,\tau)$. For a standard normal $N$, the representation under the reference measure is
\begin{equation}
\label{eq:normalrepresentation}
    \log Z=\tfrac12\log(1+z)-\tfrac12m^2\tau
          -m\sqrt{\tau(1+z)}N-\tfrac12zN^2,
    \quad M=m+P\sqrt{\frac{\tau}{1+z}}N.
\end{equation}
The Gaussian identity $\E[e^{aN-bN^2/2}]=(1+b)^{-1/2}e^{a^2/(2(1+b))}$, valid for $b>-1$, yields \eqref{eq:moments}, and the same completion of the square yields the tilted mean in \eqref{eq:mixedmoments}. The distributions in \eqref{eq:Mlaws} follow from \eqref{eq:terminalM}, and taking the $\Mprob$-expectation of $\log Z$ in \eqref{eq:Z} gives \eqref{eq:kappa}.
\end{proof}

\subsection{Terminal-budget equivalence}
\label{proof:budget}

\begin{proof}[Proof of Proposition~\ref{prop:budget}]
For an admissible strategy, discounted wealth is a square-integrable $\Mprob$-martingale, so $\E_{\Mprob}[X_T-w]=x-w=g$, and the fixed-shift argument for the exponential moment puts its terminal error in $\Cad_\eps(g)$.

Conversely, $G\in\Cad_\eps(g)$ has all polynomial moments under $\Prob$, and
\[
    \E_{\Mprob}[G^2]=\E_{\Prob}[ZG^2]
    \le(\E_{\Prob}[Z^2])^{1/2}(\E_{\Prob}[G^4])^{1/2}<\infty.
\]
By the Brownian martingale representation theorem under $\Mprob$,
\[
    \E_{\Mprob}[G\mid\F_s]=g+\int_t^s\xi_q\,dY_q,\qquad
    \E_{\Mprob}\Bigl[\int_t^T\xi_q^2\,dq\Bigr]<\infty,
\]
with a unique integrand $\xi$.
Setting $u=\xi/\sigma$, we obtain a wealth process with initial value $w+g=x$ that satisfies \eqref{eq:replication} and the expected trading-energy condition, while the fixed-shift moment bound supplies the terminal integrability condition. The strategy is unique up to $ds\otimes d\Mprob$-null sets because the Brownian integrand is.
\end{proof}

\section{Dynamic implementation}
\label{app:dynamic-proofs}

\subsection{Continuation}
\label{proof:continuation}

\begin{proof}[Proof of Theorem~\ref{thm:continuation}]
The inverse likelihood factorizes as $Z_{t,T}=Z_{t,s}Z_{s,T}$, so
\[
    G^\star=I_\eps((cZ_{t,s})Z_{s,T}),\qquad
    \E_{\Mprob}[G^\star\mid\F_s]=X_s^\star-w.
\]
By the conditional version of Theorem~\ref{thm:exact} and uniqueness of the budget root, $cZ_{t,s}$ is the current multiplier and $G^\star$ is conditionally optimal. Admissibility of the original strategy implies
\[
    \E_{\Mprob}\!\left[\int_s^T\sigma^2(u_q^\star)^2\,dq
                        \,\middle|\,\F_s\right]<\infty
    \quad\text{a.s.}
\]
The exponential bound \eqref{eq:expbound} likewise supplies the required conditional terminal moments. Evaluating the conditional objective yields \eqref{eq:Umartingale}, and continuity allows compatible versions of the processes.
\end{proof}

\subsection{Portfolio and distortion formulas}
\label{proof:feedback}

\begin{proof}[Proof of Theorem~\ref{thm:feedback}]
Write $\Gamma=1+2\eps(G^\star)^2$, $f=e^{\eps(G^\star)^2}$, and $U=\E_{\Prob}[f]$. Throughout, we differentiate the terminal payoff with the initial multiplier and initial state held fixed.
From \eqref{eq:Z}, $\partial_y\log Z=-M$, so
\begin{equation}
\label{eq:Hy}
    \partial_yG^\star=-\frac{MG^\star}{\Gamma}.
\end{equation}
Under $\Mprob$, the conditional price is the heat semigroup applied to the terminal function $G^\star(Y_T-Y_t)$, and its Brownian coefficient is the conditional expectation of the derivative in \eqref{eq:Hy}. Comparing with $d(X_s^\star-w)=\sigma u_s^\star dY_s$ yields \eqref{eq:uexact}. The bound $|G^\star/\Gamma|\le1/(2\sqrt{2\eps})$ and the Gaussian first moment of $M$ justify differentiating under the expectation. By the continuation theorem, the same calculation applies at each current state.

For $P>0$, $G^\star$ and $f$ are functions of $M$ with
\[
    \partial_MG^\star=-\frac{MG^\star}{P_T\Gamma},\qquad
    \partial_Mf=-\frac{M}{P_T}(1-\Gamma^{-1})f.
\]
With the original terminal function held fixed, the conditional distribution of $M$ under $\Prob$ has mean $m_s$ and variance $P_s-P_T$, so the diffusion coefficient of $\E_{\Prob}[f\mid\F_s]$ is $P_s\E_{\Prob}[\partial_Mf\mid\F_s]$. Dividing by that positive martingale and using $P_s/P_T=1+P_s(T-s)$, we obtain
\[
    \varphi^\star=-(1+z)\frac{\E_{\Prob}[M(1-\Gamma^{-1})f]}{U},
\]
which is \eqref{eq:phiexact} because $f/U=d\Qprob^\star/d\Prob$ and $1-\Gamma^{-1}=2\eps(G^\star)^2/\Gamma$. For $P=0$, we differentiate the Gaussian conditional expectation in the observed return instead. Then $\partial_y f=-m(1-\Gamma^{-1})f$, which leads to the same formula with $z=0$,
and the bounds in \eqref{eq:expbound} justify the differentiation.

Finally, the density process is the conditional expectation of the terminal density \eqref{eq:Qstar}, hence a true martingale.
Its entropy is finite by Theorem~\ref{thm:exact}, so the finite-entropy Girsanov theorem~\cite{leonard2012} gives finite expected energy.
\end{proof}

\begin{proof}[Proof of Corollary~\ref{cor:distortion}]
With $\Gamma$, $f$, and $U$ as in the previous proof, set $\chi=1-\Gamma^{-1}$, so that $\varphi^\star=-(1+z)\E_{\Qprob^\star}[M\chi]$ by \eqref{eq:phiexact}. If $g=0$, then $G^\star=0$ and $\varphi^\star=0$. If $m=0$, then $\varphi^\star=0$ because the law of $M$ under $\Qprob^\star$ is symmetric. By reflection, it remains to take $m>0$ and $g\ne0$.
For $P>0$, $M$ has mean $m$ and variance $P_Tz$ under $\Prob$. Both $f$ and $\chi$ are even functions of $M$, with $f>0$ and $0<\chi<1$, so pairing the normal density at $\pm M$ shows $0<\E_{\Qprob^\star}[M\chi]<\E_{\Qprob^\star}[M]$. Since $f'(M)=-M\chi f/P_T$, integration by parts under $\Prob$ yields $\E_{\Qprob^\star}[M]=m-z\E_{\Qprob^\star}[M\chi]$. Hence $(1+z)\E_{\Qprob^\star}[M\chi]<m$, and $\varphi^\star$ lies strictly between $-m$ and zero.
At $P=0$, \eqref{eq:phiexact} reads $\varphi^\star=-m\E_{\Qprob^\star}[\chi]$ with $0<\chi<1$, which gives the same conclusion.
\end{proof}

\section{Portfolio and value comparisons}
\label{app:comparison-proofs}

\subsection{The quadratic benchmark}
\label{proof:baseline}

We call an observation-predictable strategy square-integrable when
\[
    \E_{\Mprob}\Bigl[\int_t^T\sigma^2u_s^2\,ds\Bigr]<\infty,\qquad
    \E_{\Prob}[X_T^2]<\infty,
\]
which is the $\eps=0$ analog of Definition~\ref{def:admissible}. Since $\E_{\Mprob}[G_0^2]=A^2g^2\mu_3<\infty$, the benchmark payoff is replicable in this class, and the formula for $A$ follows from Lemma~\ref{lem:moments}.

\paragraph{The integrability condition of De~Franco, Nicolle, and Pham.} Along the benchmark,
\[
    X_s-w=\E_{\Mprob}[G_0\mid\F_s]=Ag\,Z_{t,s}\,\mu_2(s),
\]
with $\mu_2(s)$ evaluated at $(s,m_s,P_s)$, and the position is $u_0(s,X_s-w,m_s,P_s)$. If $P>0$, \eqref{eq:Zposterior} and Lemma~\ref{lem:moments} give
\[
    (X_s-w)^2\le C\exp\!\left(-\frac{m_s^2}{P_s(1+2z_s)}\right),
    \qquad z_s=P_s(T-s),
\]
with $C$ independent of $s$, so $m_s^2(X_s-w)^2$ is bounded; if $P=0$, $X_s-w$ is lognormal with moments bounded on $[t,T]$. Either way, \eqref{eq:u0} gives the trading-energy condition $\E_{\Prob}[\int_t^T\sigma^2u_0^2\,ds]<\infty$ of~\cite{defranco2018}.

\subsection{Position comparison}
\label{proof:attenuation}

\begin{proof}[Proof of Theorem~\ref{thm:attenuation}]
Write $\Gamma=1+2\eps(G^\star)^2$. For the non-robust payoff $G_0=AgZ$, the identity $\partial_yG_0=-MG_0$ and conditional replication yield
\[
    \sigma u_0=-Ag\,\E_{\Prob}[MZ^2]=-\frac{mg}{1+2z}.
\]
By reflection and oddness, it suffices for the portfolio comparison to take $g>0$ and $m>0$.

First, suppose $P>0$. Under $\Mprob$, $M$ is normal with positive mean $m/(1+z)$ and variance $P^2\tau/(1+z)^2=P_Tz/(1+z)$. The function $G^\star(M)>0$ is even by \eqref{eq:Zposterior}. Define
\[
    K=\E_{\Mprob}[MG^\star/\Gamma],\qquad K'=\E_{\Mprob}[MG^\star].
\]
Pairing the normal density at $M=r$ and $M=-r$ for $r>0$ shows $0<K<K'$,
because the density difference is positive and $0<G^\star/\Gamma<G^\star$ is even.
Gaussian integration by parts with $\partial_MG^\star=-MG^\star/(P_T\Gamma)$ and $\E_{\Mprob}[G^\star]=g$ then yields
\[
    K'=\frac{mg}{1+z}-\frac{z}{1+z}K.
\]
Hence $mg=(1+z)K'+zK>(1+2z)K$, and \eqref{eq:attenuation} follows from $\sigma u^\star=-K$.
At $P=0$, we have $M=m$ and $0<\E_{\Mprob}[G^\star/\Gamma]<\E_{\Mprob}[G^\star]=g$, which leads to the same conclusion. Symmetry handles the other signs and the zero states.
\end{proof}

\paragraph{Limits at maturity.} Fix $(g,m,P)$ and let $t\uparrow T$, so that $\tau\to0$ and $z\to0$. By \eqref{eq:Z} and Lemma~\ref{lem:moments}, $Z\to1$ and $M\to m$ in every $L^p(\Prob)$. The budget map $c\mapsto\E_{\Prob}[ZI_\eps(cZ)]$ therefore converges to $I_\eps$ locally uniformly, the root converges to the solution of $I_\eps(c)=g$, and $G^\star\to g$ in probability, with $|G^\star|\le|c|Z$. Dominated convergence in \eqref{eq:uexact} yields $u^\star\to-mg/(\sigma(1+2\eps g^2))$, while $u_0\to-mg/\sigma$ by \eqref{eq:u0}, and dividing gives \eqref{eq:terminalselectors}.

\subsection{Value bounds}
\label{proof:bounds}

\begin{proof}[Proof of Proposition~\ref{prop:bounds}]
The lower bound $Ag^2$ follows from Jensen's inequality and \eqref{eq:baseline}, and the upper bound from the no-trade claim $G=g$. If the adversary chooses $\Mprob$, then every admissible $G$ satisfies
\[
    J_\eps(G,\Mprob)=\E_{\Mprob}[G^2]-\frac{\kappa}{\eps}\ge g^2-\frac{\kappa}{\eps}.
\]
The objective is invariant under $(g,G)\mapsto(-g,-G)$, and reflecting $(m,y)$ leaves the distribution of $Z$, and hence the value, unchanged. At $g=0$, the claim $G=0$ attains zero.

The admissible class depends on $\eps$, so we first minimize over the larger common set $\{G\in L^2(\Mprob):\E_{\Mprob}[G]=g\}$, allowing an infinite exponential objective. The supporting-line argument of Theorem~\ref{thm:exact} still produces the same optimizer, which has every exponential moment, so the enlargement leaves the value unchanged. Each fixed-claim entropic cost is nondecreasing in $\eps$ (as is clear from the supremum over measures), and hence so is their infimum.

For the limit at zero, assume $g\ne0$ and define the truncated terminal claims
$\widetilde G_n=\sgn(g)\min(|G_0|,n)$,
$G_n=g\widetilde G_n/\E_{\Mprob}[\widetilde G_n]$.
These are bounded feasible claims with $G_n\to G_0$ in $L^2(\Prob)$, and
$\limsup_{\eps\downarrow0}V^\eps\le\E_{\Prob}[G_n^2]$ for each $n$.
Letting $n\to\infty$ and using the lower bound proves the first limit, even when $G_0$ has infinite entropic cost at every positive $\eps$. Dividing \eqref{eq:valuebounds} by $g^2$ yields the ratio limit $V^\eps/g^2\to1$, and the argument below proves the sharper constant-order limit in \eqref{eq:limitsvalue}. Finally, substituting $G=gK$ in the budget problem proves \eqref{eq:scaling}.
\end{proof}

\begin{proof}[Proof of the large-gap limit in \eqref{eq:limitsvalue}]
By the lower bound in \eqref{eq:valuebounds}, $g^2-V^\eps\le\kappa/\eps$ for every $g$. For the reverse inequality we use claims whose Gibbs measure is close to $\Mprob$: let $n>0$, set $\ell_n=\max\{-n,\min\{\log Z,n\}\}$, and put $k_n=\ell_n-\E_{\Mprob}[\ell_n]$. For $g\ne0$, the bounded claim $G=g+k_n/(2\eps g)$ has $\E_{\Mprob}[G]=g$, so $G\in\Cad_\eps(g)$, and
\[
    \eps G^2=\eps g^2+k_n+\frac{k_n^2}{4\eps g^2}.
\]
Hence
\[
    V^\eps\le g^2+\frac{1}{\eps}\log\E_{\Prob}[e^{k_n}]
         +\frac{\Vert k_n\Vert_\infty^2}{4\eps^2g^2},
\]
and therefore
\[
    \limsup_{|g|\to\infty}(V^\eps-g^2)\le\frac{1}{\eps}\bigl(\log\E_{\Prob}[e^{\ell_n}]-\E_{\Mprob}[\ell_n]\bigr).
\]
As $n\to\infty$, dominated convergence yields $\E_{\Prob}[e^{\ell_n}]\to\E_{\Prob}[Z]=1$, using $e^{\ell_n}\le1+Z$, and $\E_{\Mprob}[\ell_n]\to\E_{\Mprob}[\log Z]=\kappa$, using $|\ell_n|\le|\log Z|$ and the fact that $\log Z$ is a quadratic polynomial of a Gaussian variable under $\Mprob$. Thus $\limsup_{|g|\to\infty}(V^\eps-g^2)\le-\kappa/\eps$.
\end{proof}

\section{General priors and benchmark fragility}
\label{app:prior-proofs}

\subsection{Gaussian and known-drift benchmark tails}
\label{proof:fragility}

\begin{proof}[Proof of Theorem~\ref{thm:fragility}]
For $P>0$, substituting \eqref{eq:Zposterior} into $G_0=AgZ$ and using $M^2\ge0$ gives the deterministic bound
\begin{equation}
\label{eq:baselinebound}
    G_0^2\le g^2\frac{1+2z}{1+z}
             \exp\!\left(\frac{m^2}{P(1+2z)}\right).
\end{equation}
For $P=0$,
\[
    G_0=g\exp\!\left(-m(\Wh_T-\Wh_t)-\frac32m^2\tau\right).
\]
Its square is a nonconstant lognormal variable, so in the Gaussian integral for $\E_{\Prob}[e^{\eps G_0^2}]$ the exponential of an exponential overwhelms the quadratic Gaussian tail and the integral diverges.

The divergence also holds in the original game over measures. Put
$d\Qprob_n/d\Prob=e^{\eps\min(G_0^2,n)}/U_n$.
These densities are bounded and strictly positive, so each $\Qprob_n$ has finite entropy,
and $\E_{\Qprob_n}[G_0^2]<\infty$ by the lognormal second moment.
The objective at $\Qprob_n$ is at least $\eps^{-1}\log U_n$, which tends to infinity by monotone convergence.
\end{proof}

\subsection{General-prior integrability}
\label{proof:generalprior}

\begin{lemma}[Likelihood integrability for a general prior]
\label{lem:generalprior}
Under the hypotheses of the general-prior subsection of Section~\ref{sec:fragility}, the likelihood \eqref{eq:generallikelihood} is finite, positive, and continuous in $y$, and $\E_{\Mprob}[L_\tau]=1$. Its inverse has every positive moment under the reference measure $d\Prob=L_\tau\,d\Mprob$. Consequently the terminal optimizer and saddle of Theorem~\ref{thm:exact} remain valid for this $Z$.
\end{lemma}

\begin{proof}[Proof of Lemma~\ref{lem:generalprior}]
Completing the square bounds the exponent by $y^2/(2\tau)$ uniformly in $\rho$, which proves finiteness and gives a uniform bound on compact observation sets, so continuity follows by dominated convergence. Positivity is immediate, and $\E_{\Mprob}[L_\tau]=1$ by Tonelli. Choose $R<\infty$ with $p=\nu([-R,R])>0$; then
\[
    L_\tau(y)\ge p e^{-R|y|-R^2\tau/2}.
\]
For $n\ge1$,
$\E_{\Prob}[Z^n]=\E_{\Mprob}[L_\tau^{1-n}]$ is at most a constant times $\E_{\Mprob}[e^{(n-1)R|y|}]<\infty$, and for $0<n<1$ we use concavity and $\E_{\Prob}[Z]=1$. The proofs of existence, exponential integrability, optimality, and the saddle in Theorem~\ref{thm:exact} use only the positivity of $Z$, its unit expectation, these moments, and terminal replication, all of which hold here.
\end{proof}

\subsection{Bounded inverse likelihood}
\label{proof:boundedprior}

\begin{proof}[Proof of Proposition~\ref{prop:boundedprior}]
An atom at zero yields $L_\tau(y)\ge\nu(\{0\})$.
If both half-lines carry positive prior mass, choose in each a bounded interval of positive mass separated from zero. The positive interval forces the likelihood to infinity as $y\to\infty$ and the negative one as $y\to-\infty$, and a positive continuous function with these two limits has a positive global minimum, so its inverse is bounded.

If neither condition holds, the prior is concentrated on one open half-line. For a positive prior, the integrand in \eqref{eq:generallikelihood} tends to zero as $y\to-\infty$ and is bounded by one for $y\le0$, so $L_\tau(y)\to0$ by dominated convergence; reflection treats a negative prior. Since the distribution of $y$ under $\Mprob$ has a positive density on every open interval and the likelihood is continuous, the inverse is essentially unbounded.
\end{proof}

\subsection{Finite-support classification}
\label{proof:finiteprior}

\begin{proof}[Proof of Theorem~\ref{thm:finiteprior}]
If zero belongs to the interval, then either a support point is zero or there are support points of both signs, and Proposition~\ref{prop:boundedprior} applies.
Otherwise, after a reflection if necessary, all support points are positive.
Let $\rho_{\min}=\min_i\rho_i>0$; for $y<0$ there are positive constants $c_1,c_2$ with $c_1e^{\rho_{\min}y}\le L_\tau(y)\le c_2e^{\rho_{\min}y}$.
With $C=\eps g^2/(\E_{\Prob}[Z^2])^2>0$,
\begin{align*}
    \E_{\Prob}[e^{\eps G_0^2}]
    &=\int_\R L_\tau(y)e^{C/L_\tau(y)^2}
                    \frac{e^{-y^2/(2\tau)}}{\sqrt{2\pi\tau}}dy\\
    &\ge \frac{c_1}{\sqrt{2\pi\tau}}
       \int_{-\infty}^0
       \exp\!\left(\rho_{\min}y+\frac{C}{c_2^2}e^{-2\rho_{\min}y}-\frac{y^2}{2\tau}\right)dy
    =\infty.
\end{align*}
The last integral diverges because the exponential term in the exponent dominates the quadratic tail.
\end{proof}

\section{Precommitment target selection}
\label{app:target-proofs}

\subsection{Target embedding}
\label{proof:embedding}

\begin{proof}[Proof of Proposition~\ref{prop:embedding}]
Choose an exponent between $\eps$ and $\delta$ that leaves enough room in the fixed-shift bound for $(X-v)^2$. On bounded $v$-sets it controls the first two differentiated integrands, including their polynomial factors, so we may differentiate under the expectation to obtain \eqref{eq:Fderivatives}. Since $F_X''\ge2$, the function $F_X(v)-2\lambda v$ is coercive and has a unique minimizer $v_*$.

Define
\[
    \Phi(\Qprob,v)=\E_{\Qprob}[(X-v)^2]-2\lambda v+\lambda^2
                            -\eps^{-1}\KL(\Qprob\Vert\Prob).
\]
By Lemma~\ref{lem:gibbs} at $v=v_*$, $\Phi(\Qprob,v_*)\le\Phi(\Qprob_{v_*},v_*)$.
The first-order condition $v_*=\E_{\Qprob_{v_*}}[X]+\lambda$ says that $v_*$ also minimizes $v\mapsto\Phi(\Qprob_{v_*},v)$, so $(\Qprob_{v_*},v_*)$ is a saddle.
For any fixed $\Qprob$, completing the square yields
\[
    \min_v\Phi(\Qprob,v)=\Var_{\Qprob}(X)-2\lambda\E_{\Qprob}[X]
                                    -\eps^{-1}\KL(\Qprob\Vert\Prob).
\]
Combined with the saddle, this proves \eqref{eq:embedding} and attainment directly. Every second moment involved is finite by the entropy inequality applied to $X^2$.
\end{proof}

\subsection{Marginal cost of the gap}
\label{proof:curvature}

\begin{proof}[Proof of Proposition~\ref{prop:curvature}]
Write $U=U^\eps$. The inverse is smooth, and on bounded $c$-intervals its derivatives are dominated by polynomial powers of $Z$, also after multiplication by $e^{\eps (G^\star)^2}$ by \eqref{eq:expbound}; all these moments are finite.
Differentiating the budget gives $g_c=\E_{\Prob}[Z^2 I_\eps'(cZ)]>0$, and
differentiating the objective with $G^\star e^{\eps (G^\star)^2}=cZ$ yields $U_c=2\eps c g_c$. Hence $U_g=2\eps c$ and $V_g^\eps=U_g/(\eps U)=2c/U$.
For $c\ne0$, set $\Gamma=1+2\eps(G^\star)^2$ and $a=\E_{\Qprob^\star}[(G^\star)^2/\Gamma]$. Then
\[
    G^\star_c=\frac{G^\star}{c\Gamma},\qquad
    g_c=\frac{Ua}{c^2},
\]
and a second differentiation yields
\[
    V_{gg}^\eps=\frac{2c^2}{U^2}\left(\frac1a-2\eps\right)>0,
\]
where the sign follows from $0<a<1/(2\eps)$. At zero,
\[
    g=\mu_2c+O(c^3),\qquad U=1+\eps\mu_2c^2+O(c^4),
\]
so $V_{gg}^\eps=2/\mu_2=2A$ at $g=0$.

For the slope bounds take $g,c>0$ and put $f=e^{\eps (G^\star)^2}$. The budget and first-order condition imply
\[
    c=\E_{\Prob}[G^\star f],\quad cg=\E_{\Prob}[(G^\star)^2f],\quad
    c^2\mu_2=\E_{\Prob}[(G^\star)^2f^2].
\]
By Cauchy--Schwarz, $c^2\le cgU$, hence $c/U\le g$.
Also
\[
    c\mu_2-gU=\frac1c\Cov_{\Prob}((G^\star)^2f,f)\ge0,
\]
because both factors are increasing functions of $(G^\star)^2$; the independent-copy representation of the covariance gives this inequality without any distributional assumption. Hence $c/U\ge Ag$. If $A<1$, then $Z$ is not a.s.\ constant, since $\mu_2=1/A>1=(\E_{\Prob}[Z])^2$. Because $G^\star=I_\eps(cZ)$ is strictly increasing in $Z$, the variable $(G^\star)^2$ is not a.s.\ constant either, both factors are strictly increasing in it, and the covariance is strictly positive, so $c/U>Ag$.
\end{proof}

\subsection{Optimal target and mean--variance saddle}
\label{proof:target}

\begin{proof}[Proof of Theorem~\ref{thm:target}]
Apply Proposition~\ref{prop:embedding} to each admissible strategy and minimize jointly over the strategy and the target $v$. Since two infima always commute, translation invariance and the terminal-budget equivalence give \eqref{eq:MVvalue} after setting $g=x-v$.
The objective is coercive because $V^\eps\ge Ag^2$ and strictly convex by Proposition~\ref{prop:curvature}, and its derivative at zero is $2\lambda>0$, so its unique minimizer is negative and solves \eqref{eq:targetFOC}.
Applying \eqref{eq:slopebounds} at $|g_\lambda^\star|$ yields \eqref{eq:targetbounds}, with strict upper bounds when $A<1$ because the lower slope bound is then strict.

From now on, $G^\star$ is computed at the selected target $w_\lambda^\star$. Put $X^\star=w_\lambda^\star+G^\star$ and let $\Qprob^\star$ be its Gibbs measure for this target. By \eqref{eq:Vg} and \eqref{eq:weightedFOC},
\[
    \E_{\Qprob^\star}[G^\star]=\frac{c}{U}=-\lambda.
\]
The selected target therefore also solves the fixed-$X^\star$ target problem of Proposition~\ref{prop:embedding}, and $\Qprob^\star$ is worst case for $X^\star$ in the mean--variance criterion. Any other admissible wealth $X=X^\star+\Delta$ has $\E_{\Mprob}[\Delta]=0$, and therefore
\[
    \E_{\Qprob^\star}[G^\star\Delta]=\frac{c}{U}\E_{\Mprob}[\Delta]=0.
\]
Since $X^\star-\E_{\Qprob^\star}[X^\star]=G^\star+\lambda$,
\begin{align*}
    &\Var_{\Qprob^\star}(X)-2\lambda\E_{\Qprob^\star}[X]
       -\Var_{\Qprob^\star}(X^\star)+2\lambda\E_{\Qprob^\star}[X^\star]\\
    &\hspace{25pt}=\Var_{\Qprob^\star}(\Delta)
                        +2\E_{\Qprob^\star}[G^\star\Delta]
    =\Var_{\Qprob^\star}(\Delta)\ge0.
\end{align*}
The entropy penalty is the same on both sides and cancels, which proves the fixed-measure saddle.
Finally, Theorem~\ref{thm:attenuation} and the linearity of $u_0$ in $g$ imply
\[
    |u^\star(g_\lambda^\star)|\le|u_0(g_\lambda^\star)|
                             \le|u_0(g_\lambda^0)|,
\]
and the first inequality is strict when $m\ne0$, because $g_\lambda^\star\ne0$.
\end{proof}

\begingroup
\small
\setlength{\bibsep}{3pt}
\bibliographystyle{plainnat}
\bibliography{references}
\endgroup
\end{document}